\documentclass[a4paper,10pt]{article}

\usepackage{fullpage}
\usepackage{graphicx}
\usepackage[T1]{fontenc}
\usepackage{epsfig}
\usepackage[noadjust]{cite}
\usepackage{color}
\usepackage{url}
\usepackage{comment}
\usepackage{mathtools}
\usepackage{breqn}
\usepackage{amsmath,amssymb,amstext, amsthm}
\usepackage{amsbsy,bm}
\usepackage{amsfonts}
\usepackage{multirow}
\usepackage{enumerate}
\usepackage{hyperref}
\usepackage{lineno}
\newcommand{\ceil}[1]{\left \lceil #1 \right \rceil}

\newcommand{\rank}{\textsf{rank}}
\newcommand{\select}{\textsf{select}}
\newcommand{\rmq}{\textsf{rmq}}

\newtheorem{theorem}{Theorem}
\newtheorem{corollary}{Corollary}
\newtheorem{lemma}{Lemma}
\newtheorem{proposition}{Proposition}
\theoremstyle{definition}

\title{Encodings for Range Minimum Queries over Bounded Alphabets\footnote{Preliminary version containing some of these results has appeared in the proceedings of the 36th Annual Symposium on Combinatorial Pattern Matching (CPM 2025)~\cite{DBLP:conf/cpm/JoS25}.}}
\author{
    Seungbum Jo\footnote{This work was supported by the National Research Foundation of Korea(NRF) grant funded by the Korea government(MSIT) (RS-2025-23963814)}\\
    Chungnam National University, South Korea \\
    sbjo@cnu.ac.kr
    \and
    Srinivasa Rao Satti \\
    Norwegian University of Science and Technology, Norway \\
    srinivasa.r.satti@ntnu.no
}
\date{}
\begin{document}
\maketitle

\begin{abstract}
Range minimum queries (RMQs) are fundamental operations with widespread applications in database management, text indexing and computational biology. While many space-efficient data structures have been designed for RMQs on arrays with arbitrary elements, there has not been any results developed for the case when the alphabet size is small, which is the case in many practical scenarios where RMQ structures are used. In this paper, we investigate the encoding complexity of RMQs on arrays over bounded alphabet. We consider both one-dimensional (1D) and two-dimensional (2D) arrays. For the 1D case, we present a near-optimal space encoding. For constant-sized alphabets, this also supports the queries in constant time. For the 2D case, we systematically analyze the 1-sided, 2-sided, 3-sided and 4-sided queries and derive lower bounds for encoding space, and also matching upper bounds that support efficient queries in most cases. Our results demonstrate that, even with the bounded alphabet restriction, the space requirements remain close to those for the general alphabet case.
\end{abstract}
\section{Introduction}
Efficiently processing range queries on arrays has been one of the central problems in computer science with a wide range of applications in areas such as database management, text indexing, computational biology etc. In this paper, we focus on a few different variants of range minimum queries (RMQ) on one- and two-dimensional arrays. Given a query range within an array (that we are allowed to preprocess) a range minimum query returns the position of a smallest element within the query range.  In case of a tie, a tie-breaking rule is applied, always returning the leftmost (resp. top-leftmost) position.

Over the past decades, there has been a significant amount of research on designing space-efficient data structures that support fast queries for the RMQ problem~\cite{DBLP:journals/siamcomp/FischerH11, DBLP:conf/soda/YuanA10, DBLP:journals/algorithmica/BrodalDR12, DBLP:journals/tcs/BrodalDLRS16}. Almost all of these results focus on arrays with arbitrary elements (and a few on binary arrays~\cite{fischer2007data, DBLP:journals/talg/RamanRS07}), whereas in many real-world scenarios one often encounters input arrays that are not arbitrary but drawn from bounded or small alphabets (but not necessarily binary) - lists of rankings on a $1$–$5$ scale for example. Our aim is to investigate the effect of the alphabet size on the size of the RMQ encoding.

In the one-dimensional (1D) case, many classical data structures~\cite{HarelT84,SchieberV88,bender2000lca} have established that RMQ can be answered in constant time using linear space. The space usage is further improved to match the information-theoretic lower bound (up to lower-order terms)~\cite{DBLP:journals/siamcomp/FischerH11}. Extending the RMQ from one to two dimensions introduces new complexities. Unlike the 1D case, where the Cartesian tree provides a natural and optimal representation for encoding RMQ answers, the two-dimensional (2D) case lacks a straightforward analogue~\cite{demaine-etal}. 

In this paper we study {\em encoding data structures} for RMQ in the 1D and 2D arrays.  An encoding data structure for a given set of queries is a structure that stores sufficient information to answer those queries without directly accessing the original input data. On the other hand, an indexing data structure refers to a structure that maintains the original input while maintaining small auxiliary structures to support queries efficiently.
Our primary focus is on establishing upper and lower bounds for RMQ encodings when the input array is drawn from a bounded-sized alphabet. Many data structures, such as text indexing structures~\cite{DBLP:journals/jacm/GagieNP20, DBLP:conf/cpm/MunroNN20, DBLP:journals/algorithmica/KucherovN17, DBLP:conf/spire/FerraginaMMN04} and rank/select structures on sequences~\cite{DBLP:conf/soda/GrossiGV03, DBLP:journals/jda/Navarro14}, can be made compact and fast when the input is drawn from a bounded-sized alphabet rather than an arbitrary one. We explore whether one can exploit the bounded-alphabet input to improve the complexity of 1D and 2D RMQ encodings.


\subsection{Previous results on RMQ encodings}
\subparagraph*{1D-RMQ encodings.}
For the 1D case, any encoding for RMQ requires at least $2n - o(n)$ bits due to its bijective relationship with the Cartesian tree of the input~\cite{DBLP:journals/cacm/Vuillemin80}. The best-known result is the $(2n + o(n))$-bit data structure by Fischer and Heun that supports $O(1)$ query time.
%
When the input is highly compressible, one can design data structures whose space usage is parameterized by the compressibility of the Cartesian tree of the input~\cite{DBLP:journals/tcs/GawrychowskiJMW20, DBLP:conf/esa/MunroNBW21} while still supporting queries efficiently. Fischer considered the number of distinct Cartesian trees with bounded alphabets in his Ph.D. thesis (see Section 3.9.1 in~\cite{fischer2007data}) and computed, for each $\sigma \in \{2, 3, 4\}$, a constant $c_{\sigma}$ such that storing a Cartesian tree constructed from an input of size $n$ over an alphabet of size $\sigma$ requires at least $c_{\sigma} n - \Theta(1)$ bits.
%
\subparagraph*{2D-RMQ encodings.}
When the input is an $m \times n$ 2D array with $m \leq n$, Demaine et al.~\cite{demaine-etal} showed that there is no Cartesian tree-like structure for the 2D case. Specifically, no structure exists that fully encodes the answers to all queries that can be constructed in linear time. They further showed that when $m = n$, any encoding data structure for answering RMQ queries requires $\Omega(n^2 \log n)$ bits. Since the input array can be trivially encoded using $O(n^2 \log n)$ bits (with respect to RMQs), this implies that any encoding data structure uses asymptotically the same space as indexing data structures~\cite{DBLP:conf/soda/YuanA10, DBLP:journals/tcs/BrodalDLRS16} when $m = \Theta(n)$. Thus, most of the subsequent results focused on the case where $m = o(n)$.

Brodal et al.~\cite{DBLP:journals/tcs/BrodalDLRS16} proposed an $O(nm \cdot \min(m, \log n))$-bit encoding with $O(1)$ query time. 
Moreover, they proved that any encoding for answering RMQ requires at least $\Omega(nm \log m)$ bits. 
Brodal et al.~\cite{DBLP:conf/esa/BrodalBD13} proposed an asymptotically optimal $O(nm \log m)$-bit encoding for answering RMQ, although the queries are not supported efficiently using this encoding. 
For $m = o(n)$, the problem of designing a $o(nm \log n)$-bit data structure for a 2D array that answers queries in sublinear time remains an open problem.

For an $m \times n$ 2D array, Golin et al.~\cite{DBLP:journals/tcs/GolinIKRSS16} considered the following four type of queries: 
\begin{enumerate}
	\item $1$-sided: $[1, m] \times [1, j]$ for $1 \le j \le n$
	\item $2$-sided: $[1, i] \times [1, j]$ for $1 \le i \le m$, $1 \le j \le n$
	\item $3$-sided: $[1, i] \times [j_1, j_2]$ for $1 \le i \le m$, $1 \le j_1 \le j_2 \le n$
	\item $4$-sided: $[i_1, i_2] \times [j_1, j_2]$ for $1 \le i_1 \le i_2 \le m$, $1 \le j_1 \le j_2 \le n$ (any rectangular range).
\end{enumerate}
Golin et al.~\cite{DBLP:journals/tcs/GolinIKRSS16} provided expected upper bounds and matching lower bounds for the encoding space required to answer RMQ for these four query types assuming the elements of the input array are drawn uniformly at random.

\subsection{Our results}

\begin{table}
	\centering
	\scalebox{0.7}{
		\begin{tabular}{c | c | c | c}
			\hline
			Query type & Space (in bits) & Query time & Reference \\
			\hline
			\multicolumn{4}{c}{Upper bounds}\\
			\hline
			\multirow{2}{*}{1-sided} & $O(\log^2 n)$* & & \cite{DBLP:journals/tcs/GolinIKRSS16}* \\
			& $\min{}(\sigma, n) \ceil{\log m} + \log{{n \choose \min{}(\sigma, n)}} +o(n)$  & $O(1)$ & Theorem~\ref{thm:1-sided_gen} \\\hline
			\multirow{3}{*}{2-sided} & $O(\log^2 n \log m)$* & & \cite{DBLP:journals/tcs/GolinIKRSS16}* \\
			& $O(mn)$  & $ O(1)$ & Theorem~\ref{thm:2-sided_gen} \\
			& $\sigma \log {{n+m+2 \choose m+1}} + o(\sigma n)$  & $ O(\log  \sigma)$ & Theorem~\ref{thm:2-sided_ub} \\\hline
			\multirow{3}{*}{3-sided} & $O(n \log^2 m)$* & & \cite{DBLP:journals/tcs/GolinIKRSS16}* \\
			& $O(mn)$  & $ O(1)$ & Theorem~\ref{thm:2-sided_gen} \\
			&  $n\sigma \log m + O(\sigma n)$   & $ O(\log  \sigma)$ & Theorem~\ref{thm:3-sided_ub} \\\hline
            Column spanning 2-sided & $\ceil{n \log m} + 2n + o(n)$ & $O(1)$ & Theorem~\ref{thm:alter_2-sided_ub} \\\hline
			\multirow{4}{*}{4-sided} & $O(nm)$* & & \cite{DBLP:journals/tcs/GolinIKRSS16}* \\
			& $O(\min{}(m^2n, mn\log n))$  & $ O(1)$ & \cite{DBLP:journals/algorithmica/BrodalDR12} \\
			& $O(mn\log m)$  & $ $ & \cite{DBLP:conf/esa/BrodalBD13} \\
			&  $nm \ceil{\log \sigma} + o(mn)$   & $ O(1)$ & Theorem~\ref{thm:4-sided_ub}, $\sigma = O(1)$ \\\hline                         
			\multicolumn{4}{c}{Lower bounds}\\
			\hline
			\multirow{2}{*}{1-sided} & $\Omega(\log^2 n)$* & & \cite{DBLP:journals/tcs/GolinIKRSS16}* \\
			& $\min{}(\sigma, n) \log m + \log{{n \choose \min{}(\sigma-1, n)}}$  &  & Theorem~\ref{thm:1-sided_bounded} \\\hline
			\multirow{3}{*}{2-sided} & $\Omega(\log^2 n \log m)$* & & \cite{DBLP:journals/tcs/GolinIKRSS16}* \\
			& $\Omega(mn)$  &  & Theorem~\ref{thm:2-sided_gen} \\
			& $\Omega(\sigma m\log{\frac{n-4\sigma}{m}})$ &  & Theorem~\ref{thm:2-sided_lb}, $\sigma < n/4-4$ \\\hline
			\multirow{3}{*}{3-sided} & $\Omega(n \log^2 m)$* & & \cite{DBLP:journals/tcs/GolinIKRSS16}* \\
			&  $n\log {m-1 \choose \sigma-1} = \Omega(n \sigma \log (m/\sigma))$   &  & Theorem~\ref{thm:3-sided_lb}, $\sigma < m$ \\\hline
            \multirow{2}{*}{Column spanning 2-sided} & $n \log m$ &  & Theorem~\ref{thm:alter_2-sided_lb}, $\sigma = 2$ \\
             & $n \log m + n \log r_{\sigma-2} - O(\sigma \log n)$ &  & Theorem~\ref{thm:alter_2-sided_lb}, $\sigma  > 2$ \\\hline
			\multirow{4}{*}{4-sided} & $\Omega(nm)$* & & \cite{DBLP:journals/tcs/GolinIKRSS16}* \\
			& $\Omega(mn\log m))$  &  & \cite{DBLP:journals/algorithmica/BrodalDR12} \\
			&  $\Omega(mn\log \sigma))$   & & Theorem~\ref{thm:4-sided_lb}, $\sqrt{\sigma} \le \min{}(m, n/2)$ \\
			\hline
	\end{tabular}}
	\caption{Summary of the results of upper and lower bounds for RMQ encodings on $m \times n$ 2D arrays ($m \le n$) over an alphabet of size $\sigma$. The results marked (*) indicate expected space. $r_{\sigma-2}$ is the constant defined in Theorem~\ref{thm:1D-lowerbound}.} 
	\label{tab:summary}
\end{table}

In this paper, we study encoding data structures for RMQ on 1D and 2D arrays over a bounded-sized alphabet.
All the encoding results assume a $\Theta(\log n)$-bit word RAM model, where $n$ is the input size. 
For a 1D array of length $n$ over an alphabet of size $\sigma$, we show that any RMQ encoding needs at least $n \log (4\cos^2{(\frac{\pi}{\sigma+2})}) - O(\sigma \log n)$ bits (Theorem~\ref{thm:1D-lowerbound}). This shows that even for moderately large alphabet, the lower bound is close to the $2n-O(\log n)$-bit lower bound for the general alphabet. Moreover, we show that for any constant-sized alphabet, one can achieve optimal space usage (up to lower-order terms) while supporting the queries in constant time (Theorem~\ref{thm:1D-upperbound}). In Theorem~\ref{thm:1d1side}, we also present simple $\log {{n} \choose {\sigma-1}}$-bit lower upper bounds for 1-sided queries, where the query range is always a prefix of $[1,n]$.

For a 2D $m \times n$ array with $m \le n$, we first show (Theorem~\ref{thm:1-sided_gen}) that $\Theta(n \log m)$ bits are necessary and sufficient to answer 1-sided RMQ (in constant time).
We then generalize this bound to $\Theta(\min{(n,\sigma)}\log m)$ bits when the alphabet size is $\sigma$ (Theorem~\ref{thm:1-sided_bounded}). 
For 2-sided and 3-sided RMQ, we show that $\Theta(mn)$ bits are necessary and sufficient for answering queries in constant time (Theorem~\ref{thm:2-sided_gen} and Corollary~\ref{cor:3-sided_gen}).
On the other hand, when the input 2D array is over an alphabet of size $\sigma$, we show that any 2-sided RMQ encoding requires at least $\sigma \log {{n - 4\sigma} \choose m}$ bits (Theorem~\ref{thm:2-sided_lb}). We also show that one can match this lower bound for some ranges of parameters, by describing a data structure that takes $\sigma \log {{n+m+2} \choose m+1}+ o(\sigma n)$ bits which can answer 2-sided RMQ in $O(\log \sigma)$ time (Theorem~\ref{thm:2-sided_ub}).
For 3-sided RMQ, we show that there exists a data structure of size $n \sigma \log m + O(n \sigma)$ bits that answers queries in $O(\log \sigma)$ time (Theorem~\ref{thm:3-sided_ub}). We also show that the space bound is asymptotically optimal for some range of parameters by showing that any 3-sided RMQ encoding requires at least $n \log {{m-1} \choose {\sigma -1}}$ bits (Theorem~\ref{thm:3-sided_lb}). As a special case of 3-sided RMQ, we also consider the row/column spanning 2-sided RMQ, where the query range is restricted to $[1,m] \times [i_1,i_2]$.
We present a $(\ceil{n \log m} + 2n + o(n))$-bit data structure that answers such queries in $O(1)$ time (Theorem~\ref{thm:alter_2-sided_ub}). By Theorem~\ref{thm:alter_2-sided_lb}, this space usage is optimal up to a lower-order additive term when $m = \omega(1)$.
Finally, for the 4-sided RMQ, we first show that $\Theta(mn\log \sigma)$ bits are necessary and sufficient (Theorem~\ref{thm:4-sided_lb}), and design an encoding that supports 4-sided RMQ in constant time when $\sigma$ is $O(1)$ (Theorem~\ref{thm:4-sided_ub}).  See Table~\ref{tab:summary} for a summary of the results on 2D arrays.

The remainder of this paper is organized as follows. In Section~\ref{sec:1D}, we revisit the RMQ problem in the 1D setting over bounded-sized alphabets, establishing tight upper and lower bounds. Section~\ref{sec:2D} explores the 2D case, where we systematically consider various types of query ranges (namely, 1-, 2-, 3- and 4-sided queries) and derive upper and lower bounds on the size of the encodings, in all cases achieving asymptotically optimal bounds.

We use the following notation throughout the rest of this paper. 
Given a input 1D (resp. 2D) array and a 1D range $[i, j]$ (resp. a 2D rectangular range $[i_1, j_1] \times [i_2, j_2]$) on the array, $\rmq{}(i, j)$ (resp. $\rmq{}(i_1, j_1, i_2, j_2)$) returns the position of the smallest element within the given range.
Also for a 2D array, $(i, j)$ denotes a position at the $i$-th row and $j$-th column, and $[n]$ denotes the set $\{1, \dots, n\}$.

In addition, some encoding results in this paper use $\rank{}$ and $\select{}$ queries, defined as follows: a $\rank_\alpha(i)$ query returns the number of occurrences of $\alpha$ in the prefix $A[1,i]$, while a $\select_\alpha(j)$ query returns the position of the $j$-th occurrence of $\alpha$ in $A$, for any $\alpha \in \Sigma$ and $i,j \in {1, \dots, n}$.

\section{RMQ on 1D array over bounded-sized alphabets}\label{sec:1D}
In this section, we consider a data structure for answering RMQ on a 1D array $A[1, n]$ of size $n$ where the array elements are from an alphabet $\Sigma = \{0, \dots, \sigma-1\}$ of size $\sigma$. 
\subsection{1-sided queries} 
As a warm up, we first consider 1-sided RMQ which given an index $i$, returns the position of the minimum element in the prefix $A[1,i]$. In the following, we present a theorem that provides both upper and lower bounds for the encoding.

\begin{theorem}\label{thm:1d1side}
    Given a 1D array $A$ of size $n$ over an alphabet of size $\sigma$, at least $\log{{n \choose \sigma-1}}$ bits are necessary to answer 1-sided RMQ on $A$. Also, there exists an encoding of $\log{{n \choose \sigma-1}}$ bis to answer 1-sided RMQ on a 1D array $A$.   
\end{theorem}
\begin{proof}
For the lower bound, let $\mathcal{A}$ be the set of all arrays of size $n$ over an alphabet of size $\sigma$, constructed as follows:  
(1) pick $\sigma-1$ positions $p_{\sigma-2} < p_{\sigma-3} < \cdots < p_0$,  
(2) assign $A[p_t] = t$ for all $t \in \{0, \dots, \sigma-2\}$, and  
(3) assign $\sigma-1$ to the remaining positions.  
From this construction, the size of $\mathcal{A}$ is ${n \choose \sigma-1}$.  

Now consider two distinct arrays $A_1, A_2 \in \mathcal{A}$, and let $l$ be the leftmost position such that $A_1[l] \neq A_2[l]$.  
Without loss of generality, suppose $A_1[l] < A_2[l]$.  
Then, by construction, $\rmq{}(1,l)$ on $A_1$ returns $l$, whereas the answer to the same query on $A_2$ lies to the left of $l$. Thus, any two distinct arrays in $\mathcal{A}$ induce distinct 1-sided RMQ answers, 
which proves the lower bound part of the theorem.

For the upper bound, define a bit string $B[1,n]$ of length $n$, where for $i \in [n]$, $B[i] = 1$ if and only if $\rmq{}(1,i)$ on $A$ is $i$.  
By definition, for any $i_1 < i_2$ with $B[i_1] = B[i_2] = 1$, $A[i_2] < A[i_1]$.  
Thus, for any $j \in [n]$, the answer to $\rmq{}(1,j)$ on $A$ is simply the rightmost position of $1$ in $B[1,j]$ (or the leftmost position if no $1$ occurs in the subarray).  
Therefore, storing $B$ using $\log{{n \choose \sigma-1}}$ bits~\cite{DBLP:journals/talg/RamanRS07} gives an encoding for answering 1-sided RMQ on $A$.
\end{proof}

From the encoding of Theorem~\ref{thm:1d1side}, we can also answer 1-sided RMQ on $A$ in $O(1)$ time by using $o(n)$ bits of auxiliary structures that support $\rank{}$ and $\select{}$ on $B$ in $O(1)$ time~\cite{DBLP:journals/talg/RamanRS07}.

%
\subsection{2-sided queries} Now consider the 2-sided RMQ, where the query range may be any contiguous subarray.
We begin with the case when $\sigma =2$. 
In this case, RMQ can be answered in $O(1)$ time using an $(n + o(n))$-bit data structure~\cite{DBLP:journals/talg/RamanRS07} that supports $\rank_0$ and $\select_0$ queries on $A$ in $O(1)$ time. 
Furthermore, the following lemma shows that this data structure is optimal for answering RMQ, up to additive lower-order terms.

\begin{lemma}
	At least $n-1$ bits are necessary to answer RMQ on a 1D binary array $A$ of size $n$.  
\end{lemma}
\begin{proof}
	Consider an arbitrary binary array $A$ of length $n$ with $A[n] = 0$.
	Then for any $i \in \{1, \dots, n-1\}$, $A[i] = 0$ if and only if $\rmq{}(i, n) =  i$, since RMQ returns the leftmost position in case of a tie. Consequently, $A$ can be fully reconstructed using only RMQ, which proves the lemma.
\end{proof}

For an arbitrary 1D array $A$ of size $n$, it is known that $2n$ bits are necessary and sufficient (up to lower-order additive terms) to answer RMQ
\cite{DBLP:journals/siamcomp/FischerH11}.
The lower bound follows from the fact that if two arrays yield different RMQ results, their corresponding Cartesian trees~\cite{DBLP:journals/cacm/Vuillemin80} must have distinct shapes. 
The Cartesian tree of $A$ (denoted $C(A)$) is an unlabeled binary tree defined as follows:
(i) The root $r$ of $C(A)$ corresponds to $\rmq{}(1,n)$.
(ii) The left and right subtrees of $r$ are the Cartesian trees of 
$A[1, \rmq{}(1,n)-1]$ and $A[\rmq{}(1,n)+1, n]$, respectively, if they exist.

Let's define the {\em left height} of a binary tree $T$ as the maximum number of left edges on any root to leaf path in the binary tree.
From the definition of a Cartesian tree and the tie-breaking rule for RMQ, one can observe that the element corresponding to any node in the Cartesian tree is strictly less the element corresponding to its left child. From this we derive the following observation.

\begin{proposition}
	Given a 1D array $A$ over an alphabet of size $\sigma$, the left height of $C(A)$ 
	is at most $\sigma-1$. 
\end{proposition}

Genitrini et al.~\cite{DBLP:journals/jcta/GenitriniGKW20} study the number of compacted binary trees with bounded left height. Briefly, a compacted binary tree is a DAG representation of a binary tree obtained by repeatedly replacing identical subtrees in the original tree with pointers to their shared occurrences. Since any binary tree has the same left height as its compacted binary tree, and any two distinct unlabeled binary trees have distinct compacted representations, we directly obtain the following theorem.

\begin{theorem}[\cite{DBLP:journals/jcta/GenitriniGKW20}]\label{thm:1D-lowerbound}
	Given a 1D array $A$ of size $n$ over an alphabet of size $\sigma$, at least $n\log r_{(\sigma-1)} - O(\sigma \log n)$ bits are necessary to answer RMQ, where $r_{\sigma} = 4\cos^2{(\frac{\pi}{\sigma+3})}$.
\end{theorem}

For example, when $\sigma$ is $2, 3, 4, 5$ and $10$, at least $1$, $1.388$, $\log 3 = 1.585$, $1.698$ and $1.899$ bits per element are required to answer RMQ on $A$, respectively (see Table 1 in~\cite{DBLP:journals/jcta/GenitriniGKW20} for more examples)\footnote{for $\sigma \in \{2, 3, 4\}$, Fischer~\cite{fischer2007data} obtained the same result.}. This implies that even for constant-sized alphabets, the space lower bound for answering RMQ remains very close to that of the general case (note that for 1-sided queries, the lower and upper bounds for answering RMQ on a 1D array are asymptotically smaller than $\Theta(n)$ bits for any $\sigma = o(n)$).
From an upper bound perspective, when $\sigma = O(1)$, we can obtain an $(n\log r_{\sigma-1} + o(n))$-bit data structure that answers RMQ in $O(1)$ time, as summarized in the following theorem.

\begin{theorem}\label{thm:1D-upperbound}
	Given a 1D array $A$ of size $n$ over an alphabet of size $\sigma = O(1)$, there exists a data structure of size $n\log r_{\sigma-1} + o(n)$ bits that can answer RMQ in $O(1)$ time.
\end{theorem}
\begin{proof}
	We use a slightly modified version of the data structure from Davoodi et al.~\cite{DBLP:conf/cocoon/DavoodiRS12}. Roughly speaking, their approach decomposes $C(A)$ into $\Theta(n/\ell)$ \textit{micro-trees} of size at most  $\ell = (\log n)/4$. Each micro-tree is stored as an index into a precomputed table of $O(2^{2\ell}) = o(n)$ bits, which contains the information of all possible Cartesian trees of size at most $\ell$. One can show that these pointers into the precomputed table use $2n+o(n)$ bits. Additionally, the data structure includes $o(n)$-bit auxiliary structures to support the queries in $O(1)$ time.
	
	We modify the data structure such that the precomputed table stores only Cartesian trees of size at most $\ell$ with left height at most $\sigma-1$ using $O(2^{\ell \log r_{\sigma-1}}) = o(n)$ bits. All other $o(n)$-bit auxiliary structures in \cite{DBLP:conf/cocoon/DavoodiRS12} remain unchanged.
    As a result, the total space required to store the indexes into the precomputed table for all micro-trees is reduced to at most $\frac{n}{\ell} \cdot (\ell \log r_{\sigma-1}) = n\log r_{\sigma-1}$ bits.
\end{proof}

\section{RMQ on 2D array over bounded-sized alphabets}\label{sec:2D}
In this section, we consider a data structure for answering RMQ on a 2D array $A[1, m][1, n]$ of size $N = mn$ over an alphabet $\Sigma = \{0, \dots, \sigma-1\}$ of size $\sigma$ with $m \le n$. 

\subsection{1-sided queries}
Recall that for 1-sided RMQ, the query range is restricted to $[1, m] \times [1, i]$ for $i \in [n]$.
We begin by considering the upper and lower bounds of the data structure for the general case.

\begin{theorem}\label{thm:1-sided_gen}
	For an $m \times n$ array, (1) at least $n \log m$ bits are necessary to answer 1-sided RMQ, and 
	(2) there exists an $(n \log m + O(n))$-bit data structure that answers 1-sided RMQ in $O(1)$ time. 
\end{theorem}
\begin{proof}
	(1) Consider a set $\mathcal{B}$ containing all possible $m \times n$ 
	arrays where each array has exactly one entry with value $i$ in the $i$-th column, while all other entries are set to $n+1$. 
	The total number of such arrays is $|\mathcal{B}| = m^n$. Then, any array $B \in \mathcal{B}$ can be reconstructed using 1-sided RMQ on $B$, since the position of the value $i$ in column $i$ is given by $\rmq{}(1, m, 1, i)$. Thus, at least $n \log m$ bits are necessary to answer 1-sided RMQ.
	\\
	(2) For all indices $j$ such that $j=1$ or $\rmq{}(1, m , 1, j-1) \neq \rmq{}(1, m , 1, j)$, we store the row positions of $\rmq{}(1, m, 1, j)$ using at most $n \lceil\log m\rceil$ bits (in this case, $\rmq{}(1, m, 1, j)$ is in the $j$-th column). Additionally, we maintain a bit string $B$ of size $n$ to mark all such indices $j$, along with an auxiliary structure of $o(n)$ bits that supports $\rank{}$ and $\select{}$ queries on it in $O(1)$ time~\cite{DBLP:journals/talg/RamanRS07}. Thus, the total size of the data structure is at most $n \log m + O(n)$ bits. 
	To answer $\rmq{}(1, m, 1, i)$ in $O(1)$ time, we proceed as follows:
	(i) Identify the rightmost $j \le i$ such that $B[j] = 1$
	using $\rank{}$ and $\select{}$ queries on $B$, and 
	(ii) return $(r, j)$, where $r$ is the stored row position (of $\rmq{}(1, m, 1, j)$) corresponding to the $j$-th column.
\end{proof} 

When $A$ is defined over an alphabet of size $\sigma$, there exist at most $\min{}(\sigma, n)$ indices $j$ 
where $\rmq{}(1, m , 1, j-1) \neq \rmq{}(1, m , 1, j)$. 
Therefore, the data structure of the Theorem~\ref{thm:1-sided_gen} takes at most $\min{}(\sigma, n) \lceil\log m\rceil + \log{{n \choose \min{}(\sigma, n)}} +o(n)$ bits~\cite{DBLP:journals/talg/RamanRS07}, which is smaller than $n \log m + O(m)$ bits when $\sigma = o(n)$. Furthermore, the following theorem shows that this is optimal when $\sigma \le n$, up to additive lower-order terms (if $\sigma > n$, the data structure of Theorem~\ref{thm:1-sided_gen} gives a data structure with optimal space usage). 

\begin{theorem}\label{thm:1-sided_bounded}
	For an $m \times n$ array defined over an alphabet of size $\sigma \le n$, at least $(\sigma-1) \log m + \log {{n \choose \sigma-1}}$ bits are necessary to answer 1-sided RMQ.
\end{theorem}
\begin{proof}
	We use a similar argument as in the proof of (1) in Theorem~\ref{thm:1-sided_gen}.
	Let $n \ge j_0 > j_1 > \dots > j_{\sigma-2} \ge 1$ be $\sigma-1$ distinct column positions. 
	Define $\mathcal{B}$ as the set of all possible $m \times n$ arrays where:
	(i) For each $k \in \{0, \dots, \sigma-2\}$, exactly one entry in the $j_k$-th column has the value $k$. (ii) All other entries are set to $\sigma-1$.
	Since $j_0, \dots, j_{\sigma-1}$ can be chosen arbitrarily within the given range, 
	the total size of $\mathcal{B}$ is $m^{\sigma-1}{n \choose \sigma-1}$.
	Furthermore, any array $B \in \mathcal{B}$ can be reconstructed using 1-sided RMQ on $B$. 
	Specifically, for each $k \in \{0, \dots, \sigma-2\}$, the position of the value $k$ is given by $\rmq{}(1, m, 1, j'_k)$, where $j'_k$ is the $k$-th rightmost column such that the column position of $\rmq{}(1, m , 1, j'_k)$ is $j'_k$.
\end{proof} 

\subsection{2-sided queries}\label{sec:2d2sided}
In this section, we consider the case where the query range is restricted to $[1, i] \times [1, j]$ with $i \in [m]$ and $j \in [n]$. As with 1-sided queries, we first give the upper and lower bounds of the data structure for the general case.

\begin{theorem}\label{thm:2-sided_gen}
	For an $m \times n$ array, (1) at least $\Omega(mn)$ bits are necessary to answer 2-sided RMQ, and 
	(2) there exists an $O(mn)$-bit data structure that answers 2-sided RMQ in $O(1)$ time. 
\end{theorem}
\begin{proof}
	(1) Let $\mathcal{B}_i$ be the set of all possible 1D arrays of size $n$ such that for each $B \in \mathcal{B}_i$:
	(i) $B[1] = n\cdot i-1$, and
	(ii) for $j \in [n]\setminus \{1\}$, $B[j]$ is either $B[j-1]$ or $B[j-1]-1$. 
	The size of $\mathcal{B}_i$ is $2^{n-1}$. Furthermore, any array in $\mathcal{B}_i$ can be reconstructed using $\rmq{}(1, j)$ for all $j \in [n]$, based on the tie-breaking rule for RMQ.
	
	Next, let $\mathcal{B}$ be the set of all possible $m \times n$ arrays where the $i$-th row is an array from $\mathcal{B}_{m-i}$. The size of $\mathcal{B}$ is $2^{m(n-1)}$. Moreover, any array $B \in \mathcal{B}$ can be reconstructed using 2-sided RMQ on $B$, since the $i$-th row of $B$ can be determined from $\rmq{}(1, i, 1, j)$ for all $j \in [n]$ (note that $\rmq{}(1, i, 1, j)$ always lies in the $i$-th row by construction - since any element in the first $(i-1)$ rows is larger than every element in the $i$-th row). Thus, at least $\log |\mathcal{B}| = \Omega(mn)$ bits are necessary to answer 2-sided RMQ.
	\\\\
	(2) Given an input array $A$, define $m$ distinct 1D arrays $A_1, \dots, A_m$ of size $n$ where for each $i \in [m]$ and $j \in [n]$, $A_i[j] = \min_{k \in [i]} A[k][j]$. 
	We then construct a data structure to answer RMQ on these arrays in $O(1)$ time on each of these $m$ arrays (we do not store the arrays), using a total of $O(mn)$ bits~\cite{DBLP:journals/siamcomp/FischerH11}. Additionally, we maintain $n$ data structures for answering RMQ on each column of $A$ in $O(1)$ time, using another $O(mn)$ bits.
	To answer the 2-sided RMQ query $\rmq{}(1, i, 1, j)=(i', j')$, we proceed as follows:
	(i) Compute $j'$ in $O(1)$ time by $\rmq{}(1, j)$ on $A_i$, and (ii) compute $i'$ in $O(1)$
	time by $\rmq{}(1, i)$ on the $j'$-th column of $A$.  
\end{proof} 

Next, consider the case where the input array $A$ is defined over an alphabet $\Sigma = \{0, \dots, \sigma-1\}$ of size $\sigma$. For each position $(i, j)$, we refer to the range $[1, i] \times [1, j]$ as the 2-sided region defined by $(i, j)$.
Let $P = \{(i', j') \mid i \in [m], j \in [n], \rmq{}(1, i, 1, j) = (i', j')\}$, i.e., a set of all positions in $A$ that serve as answers to some 2-sided RMQ on $A$. 
For any two positions $(i, j)$ and $(i', j')$ of $A$, we say that $(i', j')$ \textit{dominates} $(i, j)$ if and only if $i \le i'$ and $j \le j'$. 
Because of the tie-breaking rule, it follows that for any two distinct positions $(i, j)$ and $(i', j')$ in $P$ if $(i', j')$ dominates $(i, j)$, then $A[i,j] > A[i',j']$. For example in a 2D array in Figure~\ref{figure:2sided}(a), two positions $(1, 2)$ and $(2, 3)$ are in $P$, and since the position $(2, 3)$ dominates the position $(1, 2)$, $A[1, 2] = 2$ is greater than $A[2, 3] = 0$.
We partition $P$ into set of \textit{staircases} $S_0, S_2, \dots, S_{\sigma-1}$ such that for any $\ell$,
$S_{\ell}$ is the set of all positions $(i,j)$ with $A[i,j] = \ell$. 
Then from the definition, we can directly observe the following:
\begin{proposition}\label{prop:stair}
	Consider staircases $S_1, S_2, \dots,$ within an $m \times n$ array defined over an alphabet of size $\sigma$. The following holds for any $\ell$:
	
	\begin{enumerate}[(a)]
		\item No position in $S_{\ell}$ is dominated by any other position within the same staircase.
		\item The $i$-th bottom-most position in $S_{\ell}$ coincides with the $i$-th leftmost position in $S_{\ell}$.
		\item The union of 2-sided regions defined by the positions in $S_{\ell}$ contains no value smaller than $\ell$. 
	\end{enumerate}
\end{proposition}

\begin{figure}[htp]
	\begin{center}
		\includegraphics[scale=0.28]{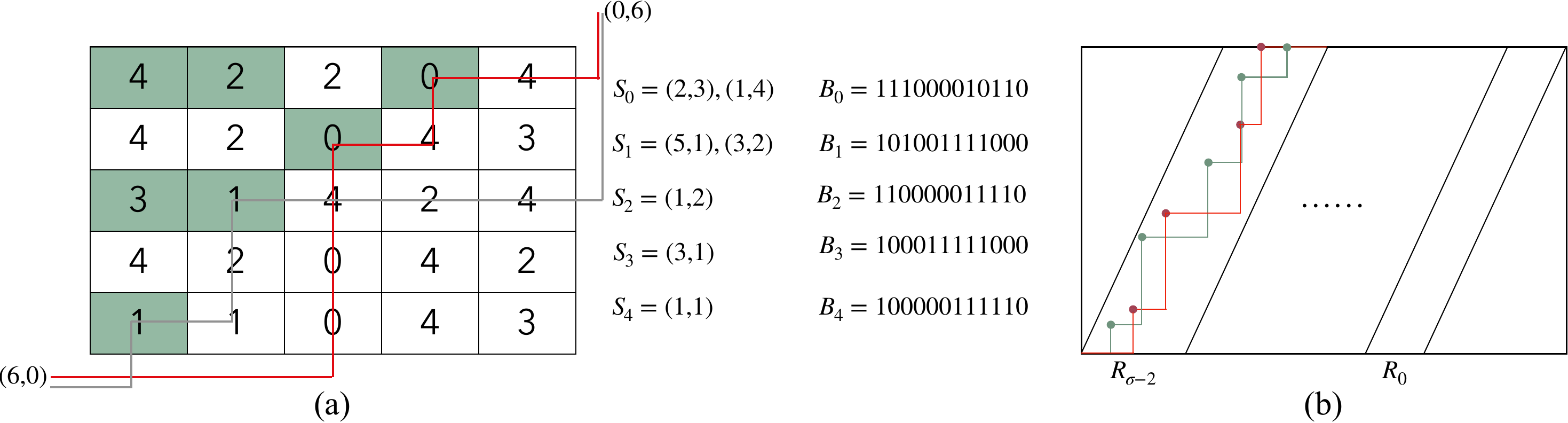}
	\end{center}
	\caption{(a) An $m \times n$ array over an alphabet of size 5, where the green positions represent the set $P$, partitioned into staircases $S_0, \dots, S_4$. The red and gray lines represent the lattice path $L_0$ and $L_1$, respectively. (b) Two distinct lattice paths in the region $R_{\sigma-2}$ and their corresponding staircases (denoted by circular points).}
	\label{figure:2sided}
\end{figure}

For each $S_{\ell}$, we define a lattice path $L_{\ell}$ from $(m+1, 0)$ to $(0, n+1)$ that satisfies the following conditions: (i) $L_{\ell}$ moves only up (i.e., from $(i, j)$ to $(i-1, j)$) or right (i.e., from $(i, j)$ to $(i, j+1)$) and passes through all positions in $S_{\ell}$ as its turning points in a clockwise order, and (ii) for any position $(i, j)$ in the array, there exists a position in $S_{\ell}$ dominated by $(i, j)$ if and only if $(i, j)$ lies on or below $L_{\ell}$. Moreover, $L_{\ell}$ can be efficiently encoded from the following lemma:

\begin{lemma}[\cite{DBLP:journals/tcs/BrodalDLRS16}]\label{lem:2-sided}
	$L_{\ell}$ can be encoded as a bitstring $B_{\ell}$ of size $m+n+2$ containing $m+1$ ones, where $B_{\ell}[p]$ is $0$ if and only if the $p$-th move of $L_{\ell}$ is north.
	Using $O(1)$ select queries on this bitstring, one can tell whether a position  $(i, j)$ lies on, below, or above $L_{\ell}$. 
\end{lemma}

We can maintain $\sigma$ bitstrings $B_0, \dots, B_{\ell}$ 
of Lemma~\ref{lem:2-sided} using $\sigma \log {{n+m+2 \choose m+1}} + o(\sigma n)$ bits~\cite{DBLP:journals/talg/RamanRS07} that support rank and select queries in $O(1)$ time.
Given any position $(i, j)$ and $\ell \in \Sigma$, we can determine whether the value at $\rmq{}(1, i, 1, j)$ is at most $\ell$ in $O(1)$ time by Lemma~\ref{lem:2-sided} and Proposition~\ref{prop:stair}. 
Consequently, the 2-sided RMQ $\rmq{}(1, i, 1, j)$ can be answered in $O(\log \sigma)$ time 
performing a binary search to find the smallest $\ell$ where $(i, j)$ lies on or below $L_{\ell}$
After finding such $\ell$, we report $(m-i'+1, j')$ as the answer, 
where $i'$ and $j'$ denote the number of zeros and ones in $B_{\ell}$ up to the $j$-th $1$, respectively  
(i.e., $(m-i'+1, j')$ is the position in $S_{\ell}$ dominated by $(i, j)$). 
Both $i'$ and $j'$ can be found in $O(1)$ time using rank and select queries on $B_{\ell}$.
We summarize the result in the following theorem.

\begin{theorem}\label{thm:2-sided_ub}
	For an $m \times n$ array $A$ defined over an alphabet of size $\sigma$, there exist a data structure using $\sigma \log {{n+m+2 \choose m+1}} + o(\sigma n)$ bits that answers 2-sided RMQ in $O(\log \sigma)$ time.
\end{theorem}

Note that for any $\sigma = o(m/\log m)$, the data structure of Theorem~\ref{thm:2-sided_ub} uses asymptotically less space than the data structure of Theorem~\ref{thm:2-sided_gen}. 
Finally, we consider the space lower bound for answering 2-sided RMQ on a 2D array with a bounded alphabet. The following theorem implies that when $\sigma \le n/c$ for any constant $c > 4$, the data structure in Theorem~\ref{thm:2-sided_ub} achieves asymptotically optimal space usage.

\begin{theorem}\label{thm:2-sided_lb}
	For an $m \times n$ array defined over an alphabet of size $\sigma < n/4-4$, at least
	$\sigma \log {{n-4\sigma \choose m}}$ bits are necessary to answer 2-sided RMQ.
\end{theorem}
\begin{proof}
	
	For an $m \times n$ array and $n' = n-4\sigma$, we define $\sigma-1$ parallelogram regions $R_{\sigma-2}, R_{\sigma-3}, \dots, R_{0}$ from the leftmost part of the array, where each region $R_{\ell}$ is determined by the four positions:
	$(1, n'+4(\sigma-2-\ell-1)+1)$, $(1, n'+4(\sigma-2-\ell))$, $(m, 4(\sigma-2-\ell-1)+1)$, and $(m, 4(\sigma-2-\ell))$. 
	For each region $R_{\ell}$, we construct a staircase $S_{\ell}$, which consists of the positions in $R_{\ell}$ whose values are $\ell$. All other positions in the array are assigned the value $\sigma-1$. Let $\mathcal{A}$ be the set of all possible such arrays.
	
	Now, consider two distinct arrays $A_1$ and $A_2$ in $\mathcal{A}$ where the staircases $S_{\ell}$ differ. Denote them as $S^1_{\ell}$ and $S^2_{\ell}$, respectively. 
	Without loss of generality, let $(i_1, j_1)$ and $(i_2, j_2)$ be the leftmost positions in $S^1_{\ell} \setminus S^2_{\ell}$ and $S^2_{\ell} \setminus S^1_{\ell}$, respectively, with $i_1 < i_2$. In this case, $(i_1, j_1)$ cannot dominate any positions in $S^2_{\ell}$, implying that the answers to the query $\rmq{}(1, i_1, 1, j_1)$ on $A_1$ and $A_2$ are different. Consequently, at least $\log {|\mathcal{A}|}$ bits are necessary to answer the 2-sided RMQ on an $m \times n$ array.
	
	To derive a lower bound on the size of $\mathcal{A}$, observe that for each region $R_{\ell}$, we can define a lattice path $L_{\ell}$ that does not cross the region boundaries. This path starts at the bottom-left corner of $R_{\ell}$, moves only north or east, and ends at the top-right corner of $R_{\ell}$. Since each distinct lattice path corresponds to a distinct staircase, defined by the positions of its turning points (in a clockwise order), 
	we obtain at least 
	${n-4\sigma \choose m}$ possible lattice paths for each region $R_{\ell}$~\cite{DBLP:journals/dm/Sato83}. See Figure~\ref{figure:2sided}(b) for an example. 
	This implies that
	$\log {|\mathcal{A}|}$ is at least $\log {{n-4\sigma \choose m}^\sigma} = \Omega(\sigma m\log{\frac{n-4\sigma}{m}})$.
\end{proof} 

\subsection{3-sided queries}
In this section, we consider the case where the query range is restricted to $[1, i] \times [j_1, j_2]$ with $i \in [m]$ and $1 \le j_1 \le j_2 \le n$. Since a 2-sided query is a special case of a 3-sided query with $j_1 = 1$, Theorem~\ref{thm:2-sided_gen} implies that at least $\Omega(mn)$ bits are necessary to answer 3-sided queries. Furthermore, the following corollary shows that this space lower bound is also asymptotically tight for 3-sided queries.

\begin{corollary}\label{cor:3-sided_gen}
	For an $m \times n$ array $A$, there exists an $O(mn)$-bit data structure that answers 3-sided RMQ in $O(1)$ time. 
\end{corollary}
\begin{proof}
	We construct the same data structure as in the proof of Theorem~\ref{thm:2-sided_gen} (2) using $O(mn)$ bits. To answer the 3-sided RMQ query $\rmq{}(1, i, j_1, j_2)=(i', j')$ in $O(1)$ time, we first compute $j'$  by $\rmq{}(j_1, j_2)$ on $A_i$, and compute $i'$ by $\rmq{}(1, i)$ on the $j'$-th column of $A$.  
\end{proof}
Next, consider the case where the input array $A$ is defined over an alphabet $\Sigma = \{0, \dots, \sigma-1\}$ of size $\sigma$.
For each $k \in \Sigma$, define a 1D array $C_k$ of size $n$ such that 
$C_k[j] =i$ where $i$ is the smallest row index such that (i) $A[i][j] = k$ and (ii) all preceding values in the column, $A[1][j], \dots, A[i-1][j]$, are greater than $k$.
If $j$-th column does not have the value $k$, set $C_k[j]= m+1$, and if $j$-th column has the value $k$, but does not satisfy the condition (ii), set $C_k[j]= 0$.
We maintain the arrays $C_0, \dots, C_{\sigma-1}$ along with RMQ data structures that allow queries to be answered in $O(1)$ time, using a total of $n\sigma \log m + O(\sigma n)$ bits~\cite{DBLP:journals/siamcomp/FischerH11}.

Given a 3-sided range $[1, i] \times [j_1, j_2]$ on $A$ and a value $k \in \Sigma$, let $j_k$ be the result of $\rmq{}(j_1, j_2)$ on $C_k$. Then the value $k$ exists in the range if and only if $C_k[j_k] \le i$.
Thus, we can determine whether $k$ exists in the range in $O(1)$ time and compute the 3-sided RMQ in $O(\log \sigma)$ time by performing a binary search to find the smallest $k$ that exists in the range.
We summarize this result in the following theorem.

\begin{theorem}\label{thm:3-sided_ub}
	For an $m \times n$ array defined over an alphabet of size $\sigma$, there exists a data structure of size $n\sigma \log m + O(\sigma n)$ bits that supports 3-sided RMQ in $O(\log \sigma)$ time. 
\end{theorem}
Compared to the general case, the above data structure requires less space when $\sigma = o(m/\log m)$, and when $\sigma = O(1)$, it answers 3-sided RMQ in $O(1)$ time.

In the general case, the optimal space data structures for answering 2-sided and 3-sided RMQ require asymptotically the same space by Theorem~\ref{thm:2-sided_gen} and Corollary~\ref{cor:3-sided_gen}. However, when the alphabet size is bounded by $\sigma$, the data structure from Theorem~\ref{thm:2-sided_ub} uses asymptotically less space than the data structure of Theorem~\ref{thm:3-sided_ub} when $m = o(n)$.
Finally, we consider the space lower bound for answering 3-sided RMQ on a 2D array with a bounded alphabet. The following theorem implies that when $\sigma$ is in $O(m^{\epsilon})$ for some positive constant $\epsilon < 1$, the data structure in Theorem~\ref{thm:3-sided_ub} achieves asymptotically optimal space usage.

\begin{theorem}\label{thm:3-sided_lb}
	For an $m \times n$ array defined over an alphabet of size $\sigma < m$, at least $n\log {m-1 \choose \sigma-1} = \Omega(n \sigma \log (m/\sigma))$ bits are necessary to answer 3-sided RMQ.
\end{theorem}
\begin{proof}
	(1) Let $\mathcal{C}$ be the set of all possible 1D arrays of size $m$ such that for each $C \in \mathcal{C}$: (i) for $k \in \{0, \dots, \sigma-2\}$, $C$ has exactly one occurrence of $k$ at the position $i_k$, 
	(ii) $i_0 > i_1 > \dots > i_{\sigma-2} > 1$, and (iii) all other positions have the value $\sigma-1$.
	Then the size of $\mathcal{C}$ is ${m-1 \choose \sigma-1}$. 
	Next, let $\mathcal{B}$ be the set of all possible $m \times n$ arrays where each column is chosen from $\mathcal{C}$. The size of $B$ is ${m-1 \choose \sigma-1}^n$. 
	Moreover, any array $B \in \mathcal{B}$ can be reconstructed using 3-sided RMQ on $B$ since $j$-th column of $B$ can be reconstructed from $\rmq{}(1, i, j, j)$ for all $i \in [n]$, which proves the theorem. Specifically, for any $k \in \{0, \dots, \sigma-2\}$, the value $k$ appears at $B[i_k][j]$ where $i_k$ is the index of the $k$-th row from the bottom that satisfies $\rmq{}(1, i, j, j) \neq \rmq{}(1, i-1, j, j)$.  
\end{proof} 

\subparagraph*{Row/column spanning 2-sided queries. } 
As a special case of 3-sided queries, we define a 2-sided query where the range is restricted to $[1,m] \times [j_1,j_2]$ with $1 \le j_1 \le j_2 \le n$ (or the range $[i_1,i_2] \times [1,n]$ with $1 \le i_1 \le i_2 \le m$). Note that this type of query range is different from the one considered in Section~\ref{sec:2d2sided}.
In this section, we refer to this variant as the \textit{column spanning (or row spanning) 2-sided RMQ} on a 2D array.
We begin by establishing a lower bound on the space required to encode column spanning 2-sided RMQ on a 2D array in the following theorem.

\begin{theorem}\label{thm:alter_2-sided_lb}
	For an $m \times n$ array defined over an alphabet of size $\sigma$, the following holds:
    \begin{enumerate}[(1)]
    \item When $\sigma = 2$, at least $n \log m$ bits, and 
    \item when $2 < \sigma \le m$, at least $n \log m + n \log r_{\sigma-2} - O(\sigma \log n)$ bits are necessary to answer column spanning 2-sided RMQ on $A$, where $r_{\sigma-2}$ is the constant defined in Theorem~\ref{thm:1D-lowerbound}.
    \end{enumerate}
\end{theorem}
\begin{proof}
    \begin{enumerate}[(1)]
    \item Let $\mathcal{C}$ denote the set of all $m \times n$ binary arrays satisfying the following conditions: (1) each column contains at most one $0$, with all other entries equal to $1$; and (2) if a column contains a $0$, it does not appear in the first row.
    It follows that $|\mathcal{C}| = m^n$. Consider any two distinct arrays $C_1, C_2 \in \mathcal{C}$, and let $l_1$ be the leftmost column in which $C_1$ and $C_2$ differ. Then the query $\rmq(1, m, l_1, l_1)$ returns distinct answers on $C_1$ and $C_2$. Note that the column consists entirely of $1$s if and only if the query returns the position in the first row by the tie-breaking rule and construction. 
    
    \item Let $\mathcal{D}$ denote the set of $m \times n$ arrays constructed as follows:
    \begin{enumerate}
        \item Let $\mathcal{W}$ be the set of all 1D arrays of length $n$ over the alphabet $\{0, \dots, \sigma-2\}$ such that any two distinct arrays in $\mathcal{W}$ yield distinct RMQ answers.
        \item For each array $W \in \mathcal{W}$, construct all possible $m \times n$ arrays as follows: for each $i \in [n]$, the $i$-th column has exactly one entry equal to $W[i]$, and all other entries are equal to $\sigma-1$. Add all such arrays to $\mathcal{D}$.
    \end{enumerate}

    By this construction and Theorem~\ref{thm:1D-lowerbound}, $\log |\mathcal{D}| \ge n \log m + n \log r_{\sigma-2} - O(\sigma \log n)$.

    Now consider any two distinct arrays $D_1, D_2 \in \mathcal{D}$. If they are derived from distinct 1D arrays in $\mathcal{W}$, their column spanning 2-sided RMQ answers are distinct by construction. Otherwise, let $l_2$ be the leftmost column in which $D_1$ and $D_2$ differ. Then the query $\rmq(1, m, l_2, l_2)$ returns distinct answers on $D_1$ and $D_2$.      
    \end{enumerate}
\end{proof}

The following theorem presents a data structure that supports column spanning 2-sided RMQ on 2D arrays efficiently. By Theorem~\ref{thm:alter_2-sided_lb}, the space usage is optimal up to a lower-order additive term when $m = \omega(1)$, for any value of $\sigma$.

\begin{theorem}\label{thm:alter_2-sided_ub}
	For an $m \times n$ array $A$ defined over an alphabet of size $\sigma$, there exists a data structure of size $\ceil{n \log m} + 2n + o(n)$ bits that can answer column spanning 2-sided RMQ in $O(1)$ time. 
\end{theorem}
\begin{proof}
Let $W_1$ be an array of length $n$ where, for each $i \in [n]$, $W_1[i]$ is the smallest value in the $i$-th column of $A$. Also let $W_2$ be an array of length $n$ where $W_2[i] = j$ if and only if $A[j][i] = W_1[i]$ (breaking ties by choosing the smallest $j$). The data structure consists of: (1) a $(2n+o(n))$-bit structure for answering RMQ on $W_1$ in $O(1)$ time~\cite{DBLP:journals/siamcomp/FischerH11}, and (2) the array $W_2$, stored in $\ceil{ n \log m}$ bits using the data structure of Dodis et al.~\cite{DBLP:conf/stoc/DodisPT10}. 

Using this structure, we can answer $\rmq(1, m, i_1, i_2)$ in $O(1)$ time by returning the position $(j_a, i_a)$, where $i_a$ is the answer to $\rmq(i_1, i_2)$ on $W_1$, and $j_a = W_2[i_a]$.
\end{proof}

Note that the space usage of the data structure in Theorem~\ref{thm:alter_2-sided_ub} is independent to $\sigma$. However, when $\sigma = O(1)$, we can reduce the space usage of the data structure into $\ceil{n \log m} + n \log r_{\sigma-1} + o(n)$ bits while supporting the query in $O(1)$ time by Theorem~\ref{thm:1D-upperbound}. 

\subsection{4-sided queries}
In this section, we consider the case where the query range is an arbitrary rectangular region.
When the alphabet size is unbounded, there exists a data structure using $O(\min{}(m^2n, mn\log n))$ bits that answers RMQ in $O(1)$ time, while at least $\Omega(mn \log m)$ bits are necessary for answering RMQ on $A$~\cite{DBLP:journals/tcs/BrodalDLRS16}. 
A $\Theta(mn \log m)$-bit encoding is also known~\cite{DBLP:conf/esa/BrodalBD13}, although it does not support the queries efficiently. Here, we focus on the case the case where the input array $A$ is defined over a bounded alphabet $\Sigma = \{0, \dots, \sigma-1\}$ of size $\sigma$.
First, we give a lower bound on the space required for the data structure, as stated in the following theorem.

\begin{figure}[htp]
	\begin{center}
		\includegraphics[scale=0.25]{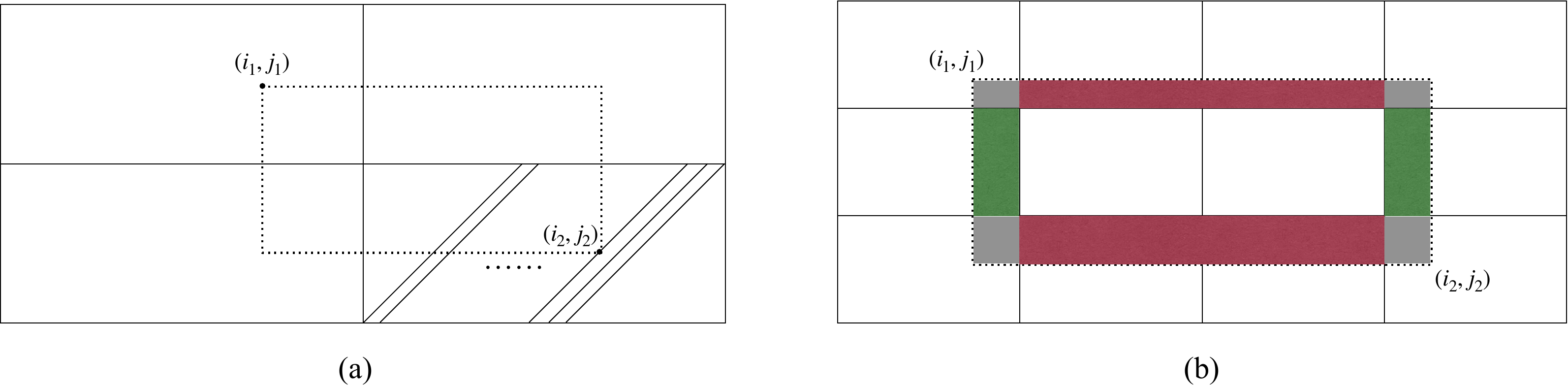}
	\end{center}
	\caption{(a) A single block of size $\sqrt{\sigma} \times 2\sqrt{\sigma}$ of an $m \times n$ array in the set $\mathcal{A}$ used in the proof of Theorem~\ref{thm:4-sided_lb}, (b) to answer RMQ in the rectangular range (represented by the dotted line): 
		(i) the candidate answer in the gray range can be found using RMQ on individual blocks, (ii) in the red range using RMQ on $A_r$, (iii) in the green range using RMQ on $A_c$, and (iv) in the white range using RMQ on $A_{rc}$.}
	\label{figure:4sided}
\end{figure}

\begin{theorem}\label{thm:4-sided_lb}
	For an $m \times n$ array defined over an alphabet $\{0, \dots, \sigma\}$ of size $\sigma+1$, 
	at least $\Omega(mn \log \sigma)$ bits are necessary to answer 4-sided RMQ, assuming $\sqrt{\sigma} \le \min{}(m, n/2)$.
\end{theorem}
\begin{proof}
	Our proof is analogous to the $\Omega(mn \log m)$-bit lower bound proof of Brodal et al.~\cite{DBLP:journals/tcs/BrodalDLRS16}. 
	First, divide an input array into blocks of size $\sqrt{\sigma} \times 2\sqrt{\sigma}$, and further divide each block into four sub-blocks of size $\frac{\sqrt{\sigma}}{2} \times \sqrt{\sigma}$. 
	Next, assign values to each block using the following procedure: 
	\begin{itemize}
		\item Assign the odd values from $\{0, \dots, \sigma-1\}$ to the upper-left sub-block in increasing order, following a row-major order.
		\item Assign the even values from $\{0, \dots, \sigma-1\}$ to the $\frac{\sqrt{\sigma}}{2}$ anti-diagonals in the bottom-right sub-block, such that the values in each anti-diagonal are larger than those in the anti-diagonals to the right. 
		\item Assign $\sigma$ to all the remaining positions in the block. 
	\end{itemize}
	
	Let $\mathcal{A}$ be the set of all $m \times n$ arrays where values are assigned according to the above procedure.
	Since, for any array $A \in \mathcal{A}$, each anti-diagonal within any sub-block of $A$ forms a permutation, 
	and there are $\frac{mn}{2\sigma}$ blocks in $A$, the total size of $\mathcal{A}$ is at least $(\sqrt{\sigma}/2)!^{\frac{\sqrt{\sigma}}{2} \cdot \frac{mn}{2\sigma}} = (\sqrt{\sigma}/2)!^{\frac{mn}{4\sqrt{\sigma}}}$. Thus, $\log {|\mathcal{A}|}$ is $\Omega(mn \log \sigma)$.
	Now, consider two distinct arrays $A_1$ and $A_2$ in $\mathcal{A}$ where $A_1[i_2, j_2] < A_2[i_2, j_2]$. From the assignment procedure, the position $(i_2, j_2)$ lies on an anti-diagonal in the bottom-right sub-block of some block (note that $A_1[i, j] = A_2[i, j]$ for every position $(i,j)$ in the other three sub-blocks of every block).
	There exists another position $(i_1, j_1)$ in the upper-left sub-block of the same block such that $A_1[i_2, j_2] < A_2[i_1, j_1] < A_2[i_2, j_2]$ (see Figure~\ref{figure:4sided}(a) for an example). 
	Thus, $\rmq{}(i_1, j_1, i_2, j_2)$ on $A_1$ returns $(i_2, j_2)$, while the same query on $A_2$ returns $(i_1, j_1)$, This implies that at least $|\mathcal{A}|$ different $m \times n$ arrays have distinct RMQ answers.
\end{proof}

Theorem~\ref{thm:4-sided_lb} implies that if $\sqrt{\sigma} \le \min{}(m, n/2)$, an $(nm \ceil{\log \sigma} + O(mn))$-bit indexing data structure (i.e., a data structure that explicitly stores the input array) of Brodal et al.~\cite{DBLP:journals/tcs/BrodalDLRS16}, which answers  RMQ in $O(1)$ time, uses asymptotically optimal space for the bounded-alphabet case.
This means that in most scenarios, restricting the alphabet size provides little advantage in terms of space efficiency compared to the general case.
However, when $\sigma = O(1)$, the $O(mn)$ term in the above indexing data structure of \cite{DBLP:journals/tcs/BrodalDLRS16} can dominate the space for storing the input. In the following theorem, we show that this can be improved to $o(mn)$ bits while still supporting $O(1)$ query time\footnote{In fact, the result of \cite{DBLP:journals/tcs/BrodalDLRS16} provides a trade-offs between index size and query time - if only $o(mn)$-bit space is used for the index, the query time increases to $\omega(1)$.}.

\begin{theorem}\label{thm:4-sided_ub}
	For an $m \times n$ array $A$ defined over an alphabet of size $\sigma = O(1)$, 
	there exists a data structure of size $mn \ceil{\log \sigma} + o(mn)$ bits that supports 4-sided RMQ in $O(1)$ time.
\end{theorem}
\begin{proof}
	Let $c = \frac{1}{2}\sqrt{\log_{\sigma} mn}$. Then we partition $A$ into blocks of size $c \times c$ and construct the following structures:
	\begin{itemize}
		\item To efficiently handle queries whose range falls entirely within a single block, we store an index data structure from Brodal et al.~\cite{DBLP:journals/tcs/BrodalDLRS16} for each block as follows. 
		We first store a precomputed table of size $2^{O(c^2 \log \sigma)} = O((mn)^{1/4})$, which contains all possible index structures from \cite{DBLP:journals/tcs/BrodalDLRS16} for blocks of size $c \times c$. 
		Since each block can be represented as an index into this precomputed table using $\frac{1}{4}\log mn$ bits, 
		the total space required for storing indices of all blocks is $\frac{mn}{4c^2} \cdot \log mn = mn\ceil{\log \sigma}$ bits. As we maintain these index structures, we can also access any position in the array in $O(1)$ time.
		\item Next, we define an auxiliary $m \times n/c$ array $A_r$, 
		where each entry is given by $A_r[i][j] = A[i][j']$, 
		where $j'$ is a column position of $\rmq{}(i, i, (j-1)c+1, jc)$ on $A$. 
		We maintain the index data structure of \cite{DBLP:journals/tcs/BrodalDLRS16} on $A_r$, using $O(\frac{mn\log \sigma}{c}) = o(mn)$ bits. This allows us to access any position in $A_r$ in $O(1)$ time and answer RMQ on $A$ in $O(1)$ time when both the leftmost and rightmost columns of the query range align with block boundaries. Note that the exact column position of the answer can be determined using the RMQ structure on individual blocks.
		
		Similarly, we define an auxiliary $m/c \times n$ array $A_c$ and maintain the same structure using $o(mn)$ bits
		so that we can answer RMQ on $A$ in $O(1)$ time when both the topmost and bottommost rows in the query range align with block boundaries. 
		
		\item  Finally, let $A_{rc}$ be an $m/c \times n/c$ array, where each entry is given by $A_{rc}[i][j] = A[i'][j']$ where $(i', j')$ is $\rmq{}((i-1)c+1, ic, (j-1)c+1, jc)$ on $A$. 
		We maintain the index data structure of \cite{DBLP:journals/tcs/BrodalDLRS16} on $A_{rc}$ using $O(\frac{mn\log\sigma}{c^2}) = o(mn)$ bits. This allows us to access any position in $A_{rc}$ in $O(1)$ time and answer RMQ on $A$ in $O(1)$ time when all boundaries of the query range align with block boundaries.
	\end{itemize}
	Any rectangular query range on $A$ can be partitioned into the following regions:
	(i) at most four rectangular regions fully contained within a single block,
	(ii) at most two rectangular regions where the leftmost and rightmost columns align with block boundaries,
	(iii) at most two rectangular regions where the topmost and bottommost rows align with block boundaries, and
	(iv) at most one rectangular region where all four boundaries align with block boundaries.
	Using the structures described above, we can determine the position and value of the minimum element in each partitioned region in $O(1)$ time.  
	See Figure~\ref{figure:4sided}(b) for an example.
	The final result is obtained in $O(1)$ time by returning the position of the minimum element among them.
\end{proof}

\section{Conclusions}\label{sec:concl}
In this paper, we studied encoding data structures for RMQ on 1D and 2D arrays when the alphabet size is bounded. Most of our data structures use asymptotically optimal space across a wide range of alphabet sizes.
We also compare our results to the general case with no restrictions on the alphabet size.
For 1D arrays, the space usage differs only slightly, even for a constant-sized alphabet. 
For 2D arrays, the data structure asymptotically requires less space compared to the general case over a wider range of alphabet sizes, from both upper and lower bound perspectives. 
For example, for an $n \times m$ 2D array with $m \le n$, we showed that the space usage can be improved for all types of queries when the alphabet size $\sigma = o(m/\log m)$.

It would be interesting to explore other parameters, such as compressed matrix size, that could further improve space usage for encoding data structures for RMQ. All of our upper and lower bounds use the fact that ties are broken using the tie-breaking rule. We note that in the case of the NLV (nearest larger value) problem (which is closely related to RMQ), Hoffman et al.~\cite{DBLP:journals/tcs/HoffmannINR18} have obtained different encoding upper and lower bounds based on different tie-breaking rules. It is interesting to see if similar results can be shown for RMQ. 
\bibliography{ref}
\end{document}